\documentclass[a4paper,reqno]{amsart}
\usepackage[text={5.2in,8in},centering]{geometry}
\usepackage{amsrefs}
\usepackage{amssymb}
\usepackage{enumerate}
\usepackage{graphicx}

\numberwithin{equation}{section}
\theoremstyle{plain}
\newtheorem{theorem}{Theorem}[section]

\newtheorem{lemma}[theorem]{Lemma}
\theoremstyle{remark}
\newtheorem{remark}[theorem]{Remark}

\newtheorem*{quest*}{Question}
\newtheorem*{remark*}{Remark}
\theoremstyle{definition}
\newtheorem{definition}[theorem]{Definition}
\newtheorem*{notation*}{Notation}
\newtheorem*{notations*}{Notations}
\providecommand{\B}{\mathbf}
\providecommand{\BS}[1]{\boldsymbol{#1}}
\providecommand{\C}{\mathcal}
\providecommand{\D}{\mathbb}
\providecommand{\F}[1]{\mathfrak{#1}}
\providecommand{\R}{\mathrm}

\newcommand{\eul}{\mathrm{e}}
\newcommand{\ii}{\mathrm{i}}

\providecommand{\prob}[1]{\D{P}\left\{ #1 \right\}}

\DeclareMathOperator{\bad}{bad}

\DeclareMathOperator{\dist}{dist}
\DeclareMathOperator{\diam}{diam}
\DeclareMathOperator{\supp}{supp}
\DeclareMathOperator{\card}{card}

\DeclareMathOperator{\expect}{\mathbb{E}}

\DeclareMathOperator{\good}{good}

\DeclareMathOperator{\intr}{int}

\DeclareMathOperator{\one}{\mathbf{1}}

\DeclareMathOperator{\out}{out}

\DeclareMathOperator{\spec}{spec}
\DeclareMathOperator{\tr}{tr}

\def\xhat{{\B{\hat{x}}}}
\def\xhatna{{\B{\hat{x}}_{n,a}}}
\def\xhatone{{\B{\hat{x}}_{n,1}}}

\begin{document}
\title[Multi-particle dynamical localization in a continuous Anderson model]
{Multi-particle dynamical localization \\ in a continuous Anderson model \\
with an alloy-type potential}
\author[V. Chulaevsky]{Victor Chulaevsky$^1$}
\author[A. Boutet de Monvel]{Anne Boutet de Monvel$^2$}
\author[Y. Suhov]{Yuri Suhov$^3$}
\address{$^1$D\'epartement de Math\'ematiques\\
Universit\'{e} de Reims, Moulin de la Housse, B.P. 1039,\\
51687 Reims Cedex 2, France\\
E-mail: victor.tchoulaevski@univ-reims.fr}
\address{$^2$Institut de Math\'ematiques de Jussieu\\
Universit\'e Paris Diderot\\
175 rue du Chevaleret, 75013 Paris, France\\
E-mail: aboutet@math.jussieu.fr}
\address{$^3$Statistical Laboratory, DPMMS\\
University of Cambridge, Wilberforce Road, \\
Cambidge CB3 0WB, UK\\
E-mail: Y.M.Suhov@statslab.cam.ac.uk}
\date{}
\begin{abstract}
This paper is a complement to our earlier work \cite{BCSS10b}. With the help of the multi-scale analysis, we derive, from estimates obtained in \cite{BCSS10b}, dynamical localization for a multi-particle Anderson model in a Euclidean space $\D{R}^{d}$, $d\geq 1$, with a short-range interaction, subject to a random  alloy-type potential.\\[2mm]
\end{abstract}
\maketitle
\section{Introduction} \label{sec:intro}
\subsection{The model}

In this paper we continue our study of a multi-particle Anderson model in $\D{R}^d$ with interaction and in an external random potential of alloy type. The Hamiltonian $\B{H}$ $\left(=\B{H}^{(N)}(\omega)\right)$ is a random Schr\"odinger operator of the form
\begin{equation}\label{eq:def.H}
\B{H}=-\frac{1}{2}\B{\Delta}+\B{U}(\B{x})+\B{V}(\omega;\B{x})
\end{equation}
acting in $L^2(\D{R}^{Nd})$. This means that we consider a system of $N$ interacting quantum particles in $\D{R}^d$. Here $\B{x}=(x_1,\ldots,x_N)\in\D{R}^{Nd}$ is for the joint position vector, where each component $x_j\in\D{R}^d$ represents the position of the $j$th particle, $1\leq j\leq N$. Next, $\B{\Delta}$ stands for the Laplacian in $\D{R}^{Nd}$. The interaction energy operator $\B{U}(\B{x})$ acts as multiplication by a function $U(\B{x})$. Finally, the term $\B{V}(\omega;\B{x})$
represents the operator of multiplication by a function
\begin{equation}\label{eq:def.V}
\B{x}\mapsto V(x_1;\omega)+ \cdots + V(x_N;\omega),
\end{equation}
where $x\in\D{R}^d\mapsto V(x;\omega)$ is a random external field potential assumed to be of the form
\begin{equation}\label{eq:defalloy}
V(x;\omega)=\sum_{s\in\D{Z}^d}\R{V}_s(\omega)\, \varphi(x-s).
\end{equation}
Here and below $\R{V}_s$, $s\in\D{Z}^d$, are i.i.d.\ (independent and identically distributed) real random variables on some probability space $(\Omega,\F{B},\D{P})$ and $\varphi\colon\D{R}^d\to\D{R}$ is usually referred to as a ``bump'' function.

\subsection{Basic geometric notations}

Throughout this paper, we will fix an integer $N \geq 2$ and work in Euclidean spaces of the form $\D{R}^{ld}\cong\D{R}^d\times\ldots\times\D{R}^d$ ($l$ times) associated with $l$-particle sub-systems where $1\leq l\leq N$. Correspondingly, the notations $\B{x}$, $\B{y}$,\,\dots\ will be used for vectors from $\D{R}^{ld}$, depending on the context. Given a vector $\B{x}\in\D{R}^{ld}$, we will consider ``sub-configurations'' $\B{x}'$ and $\B{x}''$ generated by $\B{x}$ for a given partition of an $l$-particle system into disjoint sub-systems with $l'$ and $l''$ particles, where $l'+l''=l$, $l',l''\geq 1$; the vectors $\B{x}'$ and $\B{x}''$ are identified with points from $\D{R}^{l'd}$ and $\D{R}^{l''d}$, respectively, by re-labelling the particles accordingly.

All Euclidean spaces will be endowed with the max-norm denoted by $|\,\cdot\,|$. We will consider $ld$-dimensional cubes of integer size in $\D{R}^{ld}$ centered at lattice points $\B{u}\in\D{Z}^{ld}\subset\D{R}^{ld}$ and with edges parallel to the co-ordinate axes. The cube of edge length $2L$ centered at $\B{u}$ is denoted by $\BS{\varLambda}_L(\B{u})$; in the max-norm it represents the ball of radius $L$ centered at $\B{u}$:
\begin{equation}\label{eq:BLam}
\BS{\varLambda}_L(\B{u}) =\{\B{x}\in\D{R}^{ld}:\;|\B{x} - \B{u}| < L\}.
\end{equation}
The lattice counterpart for $\BS{\varLambda}_L(\B{u})$ is denoted by $\B{B}_L(\B{u})$:
\[
\B{B}_L(\B{u}) = \BS{\overline{\varLambda}}_L(\B{u}) \cap \D{Z}^{ld};
\quad \B{u}\in\D{Z}^{ld}.
\]
Finally, we consider ``cells'' (cubes of radius $1$) centered at lattice points $\B{u}\in\D{Z}^{ld}$:
\[
\B{C}(\B{u}) = \BS{\varLambda}_1(\B{u})\subset \D{R}^{ld}.
\]
The union of all cells $\B{C}(\B{u})$,
$\B{u}\in\D{Z}^{ld}$, covers the entire Euclidean space
$\D{R}^{ld}$.
For each $i\in\{1,\ldots,l\}$ we introduce the projection $\Pi_i\colon\D{R}^{ld}\to\D{R}^{d}$ defined by
\[
\Pi_i\colon(x_1,\ldots,x_l)\longmapsto x_i,\;\;1\leq i\leq l.
\]

\subsection{Interaction potential}

The interaction within the system of particles is represented by the term $\B{U}(\B{x})$ in the expression \eqref{eq:def.H} of the Hamiltonian $\B{H}$. As was said, it is the operator of multiplication by a function $\B{x}\in\D{R}^{ld}\mapsto U(\B{x})\in\D{R}$, $1\leq l\leq N$. A usual assumption is that $U(\B{x})$ (considered for $\B{x}\in\D{R}^{ld}$
with $1\leq l\leq N$) is a sum of $k$-body potentials
\[
U(\B{x})=\sum_{k=1}^l\,\sum_{1\leq i_1<\ldots<i_k\leq l}U^{(k)}(x_{i_1},\ldots,x_{i_k}),\qquad\B{x}=(x_1,\ldots,x_l)\in\D{R}^{ld}.
\]

In this paper we do not assume isotropy, symmetry or translation invariance of this interaction. However, we use the conditions of finite range, nonnegativity and boundedness, as stated below.

Assume a partition of a configuration $\B{x}\in\D{Z}^{ld}$ is given, into complementary sub-configurations $\B{x}_{\C{J}}=(x_j)_{j\in\C{J}}$ and $\B{x}_{\C{J}^{\rm c}}=(x_j)_{j\in\{1,\ldots,l\}\setminus\C{J}}$, where $\varnothing\neq\C{J}\subsetneq\{1, 2, \ldots, l\}$. The \emph{energy of interaction} between $\B{x}_{\C{J}}$ and $\B{x}_{\C{J}^{\rm c}}$ is defined by
\begin{equation}\label{eq:int.1}
U(\B{x}_{\C{J}} \, | \, \B{x}_{\C{J}^{\rm c}}):= U(\B{x}) - U(\B{x}_{\C{J}})- U(\B{x}_{\C{J}^{\rm c}}).
\end{equation}
Next, define
\begin{equation}\label{eq:rho}
\rho(\B{x}_{\C{J}},\B{x}_{\C{J}^{\rm c}}):=\min\;\Big[|x_i - x_j|:\;i\in\C{J},j\in\C{J}^{\rm c}\Big].
\end{equation}
We say that this interaction has \emph{range} $\R{r}_0\in(0,\infty)$ if, for all $l=1,\ldots,N$ and $\B{x}\in\D{R}^{ld}$,
\begin{equation}\label{eq:int.shortr}
\rho(\B{x}_\C{J},\B{x}_{\C{J}^{\rm c}}) >\R{r}_0\implies U(\B{x}_{\C{J}}\, |\,\B{x}_{\C{J}^{\rm c}}) = 0.
\end{equation}
Finally, we say that the interaction is \emph{non-negative} and \emph{bounded} if
\begin{equation}\label{eq:int.bndd}
\inf_{\B{x}\in\D{R}^{ld}}U(\B{x})\geq 0\;\text{ and }\;\sup_{\B{x}\in\D{R}^{ld}}U(\B{x})<+\infty,\quad 1\leq l\leq N.
\end{equation}
The boundedness condition can be relaxed to include hard-core interactions where $U(\B{x})=+\infty$ if $|x_i-x_j|\leq a$, for some given $a\in (0,\R{r}_0)$.

\subsection{Assumptions}

Our assumptions on the interaction potential $U$ are borrowed from \cite{BCSS10b}:
\begin{enumerate}[\bf{(E}1)]
\item
$U$ is non-negative, bounded
and has a finite range $\R{r}_0\geq 0$.
\end{enumerate}

Similarly, we use assumptions on the i.i.d.\ random variables
$\R{V}_s$, $s\in\D{Z}^d$, and the bump function $\varphi$
introduced in \cite{BCSS10b}:
\begin{enumerate}[\bf{(E}1)]
\addtocounter{enumi}{1}
\item
There exists a constant $\R{v}\in (0,\infty)$ such that
\begin{equation} \label{eq:external.boundedness}
\prob{ 0 \leq \R{V}_0 \leq \R{v}} = 1
\end{equation}
and
\begin{equation} \label{eq:external.zero}
\forall\;\epsilon > 0\quad
\prob{\R{V}_0 \leq \epsilon} > 0.
\end{equation}
\item
\emph{Uniform H\"older continuity}:\footnote{The H\"older continuity can be relaxed to the $\log$-H\"older continuity.} There exist constants $\R{a},\,\R{b}>0$ such that for all $\epsilon\in[0,1]$, the common distribution function $F$ of the random variables $\R{V}_s$ satisfies
\begin{equation}           \label{eq:external.holder}
\sup_{\R{y}\in\D{R}}\bigl\lbrack F(\R{y}+\epsilon)-F(\R{y})\bigr\rbrack\leq\R{a}\epsilon^{\R{b}}.
\end{equation}
\item
The function $\varphi\colon\D{R}^d\to \D{R}$ is bounded, nonnegative and compactly supported:
\begin{equation}            \label{eq:compact.supp.bumps}
\diam(\supp\varphi)\leq\R{r}_1 < \infty.
\end{equation}
\item
For all $L\geq 1$ and $u\in\D{Z}^d$,
\begin{equation}       \label{eq:covering.condition}
\sum_{s\in\varLambda_L(u)\cap\D{Z}^d}\;\varphi(x-s)\geq\one_{\varLambda_L(u)}(x).
\end{equation}
\end{enumerate}
Here and below, $\one_A$ stands for the indicator function of a set $A$.

Henceforth, we suppose that $d$ and $N$ are fixed, as well as the interaction $\B{U}$ and the structure of the external potential (i.e., the distribution function $F$ and the bump function $\varphi$). All constants emerging in various bounds below are introduced under this assumption.

\subsection {Dynamical localization}

The main result of this paper, Theorem~\ref{thm:main}, establishes the so-called ``strong dynamical localization'' for the operator $\B{H}(\omega)$ defined in \eqref{eq:def.H} near the lower edge $E^0$ of its spectrum. More precisely, let $E^0$ be the lower edge of the spectrum $\spec(\B{H}^0)$ of the $N$-particle operator without interaction,
\begin{equation}\label{eq:opH0}
\B{H}^0=-\frac{1}{2}\B{\Delta}+\sum_{j=1}^N V(x_j;\omega).
\end{equation}
Actually, it follows from our conditions \eqref{eq:external.boundedness} and \eqref{eq:external.zero} that $E^0 = 0$. Owing to the non-negativity of the interaction potential $U$, the lower edge of the spectrum of $\B{H}$ is bounded from below by $E^0$. Moreover, $\B{H}$ has a non-empty spectrum in the interval $[E^0,E^0+\epsilon ]$, for any $\epsilon>0$. This follows, e.g., from a result by Klopp and Zenk \cite{KZ03} which says that the integrated density of states for a multi-particle system with a decaying interaction is the same as for the system without interaction.

Denote by $\B{X}$ the operator of multiplication by the norm of $\B{x}$, i.e.,
\begin{equation}\label{eq:opX}
\B{X}f(\B{x}) = |\B{x}| \, f(\B{x}), \quad \B{x}\in \D{R}^{Nd}.
\end{equation}
The main result of this paper is the following

\begin{theorem}                     \label{thm:main}
Consider  the operator $\B{H}$ from \eqref{eq:def.H} and assume that conditions \emph{\textbf{(E1)}--\textbf{(E5)}} are fulfilled. Then for any $Q>0$ there exists a nonrandom number $\eta =\eta(Q)>0$ such that for any compact subset $\B{K}\subset\D{R}^{Nd}$ the following bound holds:
\begin{equation}
\expect\left[ \sup_{t\in\D{R}} \;\left\| \B{X}^Q \, \eul^{-\ii t\B{H}(\omega)}
          P_{I(\eta)}(\B{H}(\omega))
\one_{\B{K}}\right\|_{L^2(\D{R}^{Nd})} \right] < \infty,
\end{equation}
where $P_{I(\eta)}(\B{H})$ is the spectral projection of the Hamiltonian $\B{H}$ on the interval $I(\eta)=[E^0,\,E^0+\eta]$.
\end{theorem}

\begin{remark}
The interval $I(\eta)$ is a sub-interval of the interval of energies $[E^0,E^0+\eta^*]$ for which the spectrum of $\B{H}$ was proven to be pure point (and the eigenfunctions to be decaying exponentially); see~\cite{BCSS10b}.
\end{remark}

\section{Results of the multi-particle MSA} \label{sec:resMSA}

The MSA works with the finite-volume approximations $\B{H}_{\BS{\Lambda}_L(\B{u})}$ of $\B{H}$, relative to the cubes ${\BS{\Lambda}_L(\B{u})}$. More precisely, $\B{H}_{\BS{\Lambda}_L(\B{u})}$ is an operator in $L^2(\BS{\Lambda}_L(\B{u}))$, given by the same expression as in \eqref{eq:def.H} (for $\B{x}\in\BS{\Lambda}_L(\B{u})$), with Dirichlet's boundary conditions on $\partial\BS{\Lambda}_L(\B{u})$; see~\cite{BCSS10b}. Specifically, the Green operator $\B{G}_{\BS{\varLambda}_L(\B{u})}(E)$ is of particular interest:
\begin{equation}\label{eq:GreenOp}
\B{G}_{\BS{\varLambda}_L(\B{u})}(E)=
(\B{H}_{\BS{\Lambda}_L(\B{u})} - E)^{-1},
\end{equation}
defined for $E\in\D{R}\setminus\spec\left(\B{H}_{\BS{\Lambda}_L(\B{u})}\right)$.

Let $[\;\cdot\;]$ denote the integer part. For a cube $\BS{\varLambda}_L(\B{u})$ we denote
\begin{equation}
\BS{\varLambda}^{\intr}_L(\B{u})=\BS{\varLambda}_{[ L/3]}(\B{u)},\quad\BS{\varLambda}^{\out}_L(\B{u})=\BS{\varLambda}_L(u) \setminus \BS{\varLambda}_{L-2}(u).
\end{equation}
Next, given two points $\B{v},\B{w}\in\B{B}_L(\B{u})$ such that $\B{C}(\B{v}), \B{C}(\B{w})\subset\BS{\varLambda}_L(\B{u})$, set
\begin{equation}\label{eq:G.u.v}
\B{G}_{\B{v}, \B{w}}^{\BS{\varLambda}_L(\B{u})}(E):=\one_{\B{C}(\B{v})}\B{G}_{\BS{\varLambda}_L(\B{u})}(E)\one_{\B{C}(\B{w})}.
\end{equation}
Following a long-standing tradition, we use a parameter $\alpha\in(1,2)$ in the definition of a sequence of scales $L_k$ (cf. Eqn~\eqref{eq:scales}); For our purposes, it suffices to set $\alpha=3/2$; this will be always assumed below.
\begin{definition}\label{DefNS}
A cube $\BS{\varLambda}_L(\B{u})$ is called \emph{$(E,m)$-non-singular} ($(E,m)$-NS, in short) if for any $\B{v}\in \B{B}_{\left[ L^{1/\alpha}\right]}(\B{u})$ and $\B{y}\in \BS{\varLambda}_L^{\out}(\B{u})\cap\D{Z}^{Nd}$ the norm of the operator $\B{G}_{\B{v},\B{y}}^{\BS{\varLambda}_L(\B{u})}(E)$ satisfies
\begin{equation}\label{eq:cubeNS}
\left\|\B{G}_{\B{v},\B{y}}^{\BS{\varLambda}_L(\B{u})}(E)\right\|_{L^2(\BS{\Lambda}_L(\B{u}))}\leq\eul^{-mL}.
\end{equation}
Otherwise, it is called \emph{$(E,m)$-singular} ($(E,m)$-S).
\end{definition}

We will work with a sequence of ``scales'' $L_k$ (positive integers) defined  recursively by
\begin{equation}\label{eq:scales}
L_k:=\left[L_{k-1}^\alpha\right] + 1,\;\text{ where }\;\alpha =\frac{3}{2}\,.
\end{equation}
The sequence $L_k$ is determined by an initial scale $L_0\geq 2$. Most of arguments in Sect.~\ref{sec:reduction} require $L_0$ to be large enough, to fulfill some specific numerical inequalities. In addition, we also assume that $L_0\geq\R{r}_1$ (defined in \eqref{eq:compact.supp.bumps}) in order to simplify some cumbersome technicalities.

We will use a well-known property of generalized eigenfunctions of the operator $\B{H}$ which can be found, e.g., in \cite[Lemma 3.3.2]{St01}:

\begin{lemma}\label{lem:EDI}
For every bounded set $I_0\subset\D{R}$ there exists a constant $C^{(0)}=C^{(0)}(I_0)$ such that, for any cube $\BS{\Lambda}_L(\B{u})$ with $L>7$, any point $\B{w}\in\B{B}_L(\B{u})$ with $\B{C}(\B{w})\subseteq\BS{\Lambda}^{\intr}_L(\B{u})$ and every generalized eigenfunction $\varPsi$ of $\B{H}$ with eigenvalue $E\in I_0$, the norm of the vector $\one_{\B{C}(\B{w})}\BS{\varPsi}$ satisfies
\begin{equation}\label{eq:EDI}
\|\one_{\B{C}(\B{w})} \BS{\varPsi}\|
\leq C^{(0)}\,\| \one_{\BS{\Lambda}^{\out}_L(u)}
\B{G}_{\BS{\varLambda}_L(\B{u})}(E) \, \one_{\B{C}(\B{w})} \|
\cdot \|\one_{\BS{\Lambda}^{\out}_L(\B{u})}  \BS{\varPsi} \|.
\end{equation}
\end{lemma}

\noindent(From now on we omit the subscript indicating the
$L^2$-space where a given norm is considered, as this will be
clear in the context of the argument.)

The following geometric notion is used in the forthcoming analysis.

\begin{definition} (see~\cite{BCSS10b}).
Let $\C{J}$ be a non-empty subset of $\{1,\ldots,N\}$.

We say that the cube $\BS{\varLambda}_L(\B{y})$ is \emph{$\C{J}$-separable} from the cube $\BS{\varLambda}_L(\B{x})$ if
\begin{equation}\label{eq:sep}
\Biggl(\,\bigcup_{j\in \C{J}} \Pi_j \BS{\varLambda}_{L+ \R{r}_1}(\B{y})
\Biggr)\bigcap\Biggl(\,\bigcup_{i\not\in\C{J}} \Pi_i \BS{\varLambda}_{L+ \R{r}_1 }(\B{y}) \;\bigcup \Pi \BS{\varLambda}_{L+ \R{r}_1}(\B{x})\Biggr) = \varnothing
\end{equation}
where $\Pi\BS{\varLambda}_{L+\R{r}_1}(\B{x})=\cup_{j=1}^N\Pi_j\BS{\varLambda}_{L+\R{r}_1}(\B{x})$.

A pair of cubes $\BS{\varLambda}_L(\B{x})$, $\BS{\varLambda}_L(\B{y})$ is \emph{separable}  if, for some $\C{J}\subseteq\{1,\ldots,N\}$, either $\BS{\varLambda}_L(\B{y})$ is $\C{J}$-separable from $\BS{\varLambda}_L(\B{x})$, or $\BS{\varLambda}_L(\B{x})$ is $\C{J}$-separable from $\BS{\varLambda}_L(\B{y})$.
\end{definition}

We will use the following easy assertion (see~\cite{BCSS10b}):

\begin{lemma}       \label{lem:CondGeomSep}
For any $L>1$ and $\B{x}\in\D{R}^{Nd}$, there exists a collection of $N$-particle cubes $\BS{\varLambda}_{2N(L+\R{r}_1)}(\B{x}^{(l)})$, $l=1,\ldots,K(\B{x},N)$, with $K(\B{x},N)\leq N^N$, such that if a vector $\B{y}\in\D{Z}^{Nd}$ satisfies\,\footnote{\,The constant $\R{r}_1$ is defined in \eqref{eq:compact.supp.bumps}.}
\begin{equation}  \label{eq:CondGeomSep}
\B{y}\notin\bigcup_{\ell=1}^{K(\B{x},N)}\BS{\varLambda}_{2N(L+\R{r}_1)}
(\B{x}^{(l)}),
\end{equation}
then two cubes $\BS{\varLambda}_L(\B{x})$ and $\BS{\varLambda}_L(\B{y})$ with $\dist\left(\BS{\varLambda}_L(\B{x}),\BS{\varLambda}_L(\B{y})\right)>2N(L+\R{r}_1)$ are separable. In particular, assuming $L\geq \R{r}_1$, a pair of cubes $\BS{\varLambda}_L(\B{x})$, $\BS{\varLambda}_L(\B{y})$ is separable if
\begin{equation}\label{eq:cond.sep.4NL}
|\B{y}| > |\B{x}| + (4N+2)L.
\end{equation}
\end{lemma}

Since $N\geq 2$, one can replace the condition \eqref{eq:cond.sep.4NL} by
\begin{equation}\label{eq:cond.sep.5NL}
|\B{y}| > |\B{x}| + 5NL.
\end{equation}
In particular, two cubes of the form $\BS{\varLambda}_L(\B{0})$, $\BS{\varLambda}_L(\B{y})$ with $|\B{y}|> 5NL$ are always separable.

The main outcome of \cite{BCSS10b} is summarized in the following Theorem~\ref{thm:DSbound}:

\begin{theorem}[see~\cite{BCSS10b}]                 \label{thm:DSbound}
For any large enough $p>0$ there exist $m^*(p)>0$, $\eta^*(p)>0$, and $L^*_0(p)>0$
such that
\begin{enumerate}[\rm(i)]
\item
if $L_0\geq L_0^*(p)$ then for all $k\geq 0$ and for any pair of separable cubes
$\BS{\varLambda}_{L_k}(\B{x})$, $\BS{\varLambda}_{L_k}(\B{y})$ with $\B{x},\B{y}\in\D{Z}^{Nd}$,
\begin{equation}\label{eq:DS}
\prob{ \exists\, E\in[E^0, E^0+\eta^*]:\,
\text{$\BS{\varLambda}_{L_k}(\B{x})$ and $\BS{\varLambda}_{L_k}(\B{y})$
are $(E,m)$-{\rm S}} } \leq L_k^{-2p},
\end{equation}
\item
with probability one, the spectrum of $\B{H}$ in the interval $I=[E^0,E^0+\eta^*(p)]$ is pure point, and the eigenfunctions $\BS{\varPhi}_n$ of $\B{H}$ with eigenvalues $E_n\in I$ satisfy
\begin{equation}\label{eq:ExpLoc}
\|\BS{\varPhi}_n\one_{\B{C}(\B{w})}\|\leq
C_n(\omega)\eul^{-m^*(p)|\B{w}|},\quad\B{w}\in\D{Z}^{Nd},\quad C_n(\omega) < \infty.
\end{equation}
\end{enumerate}
\end{theorem}

\section{Derivation of dynamical localization from MSA estimates}
\label{sec:reduction}

In this section we prove a statement that is slightly more general than Theorem~\ref{thm:main}. Namely, given $Q>0$, the interval $I = I(\eta)=[E^0,E^0+\eta]$ with $\eta=\eta(Q)$, and a compact subset $\B{K}\subset\D{R}^d$, there exists a constant $C(Q,\B{K})\in (0,\infty)$ such that for any bounded measurable function $\xi\colon\D{R}\to\D{C}$ with $\supp\xi\subset I(\eta)$,
\begin{equation}\label{eq:DL2}
\expect \left[  \| \B{X}^Q \, \xi(\B{H}(\omega))\one_{\B{K}} \| \right]
< C(Q,\B{K})\,\|\xi\|_\infty <\infty.
\end{equation}
Moreover, $Q>0$ can be made arbitrarily large, by choosing $\eta =\eta (Q)>0$ sufficiently small. Theorem~\ref{thm:main} follows from \eqref{eq:DL2} applied to the functions $\xi(s)= \eul^{-\ii ts}\one_{I(\eta)}(s)$, parametrised by $t\in\D{R}$.

Throughout the section, we assume that the parameter $p$ from \eqref{eq:DS} satisfies
\begin{equation}\label{eq:DD.2}
2p > 3Nd\alpha + \alpha Q.
\end{equation}
More precisely, given $Q>0$ and $p$ satisfying \eqref{eq:DD.2}, we work with
\begin{equation}
\eta=\eta (Q)\in (0,\eta^* (p))\;\text{ and }\;m=m^*(p)>0,
\end{equation}
where $\eta^*(p)$ and $m^*(p)$ are specified in Theorem~\ref{thm:DSbound}. Further, for $p$ satisfying \eqref{eq:DD.2} we introduce the event $\Omega_1=\Omega_1(p)\subseteq\Omega$ of probability $\D{P}(\Omega_1)=1$, defined by
\begin{equation}
\Omega_1=\bigl\{\omega\in\Omega:\;\text{the spectrum of $\B{H}(\omega)$
 in $[E^0,E^0+\eta^*(p)]$ is pure point}\bigr\}.
\end{equation}

\subsection{Probability of ``bad samples''}\label{ssec:step.1}

Given $j\geq 1$, consider the event
\begin{align*}
\C{S}_j
&=\{ \omega:\,\text{there exists }E\in I\,\text{ and }\, \B{y},\B{z}\in
\B{B}_{5NL_{j+1 }}(\B{0}) \text{ such that}\\
&\qquad\qquad\qquad\quad\BS{\Lambda}_{L_{j}}(\B{y}), \BS{\Lambda}_{L_{j}}(\B{z})
\text{ are  separable and $(m,E)$-S}\}.
\end{align*}
Further, for $k\geq 1$ we denote
\begin{equation}\label{eq:Ombad}
\Omega_k^{\bad} = \bigcup_{j\geq k} \C{S}_j.
\end{equation}

\begin{lemma}
There exists a constant $c_1\in(0,\infty)$ such that for all $k\geq 1$,
\[
\prob{\Omega_k^{\bad}}\leq c_1L_k^{-(2p-2Nd\alpha)}.
\]
\end{lemma}

\begin{proof}
The number of separable pairs $\BS{\Lambda}_{L_{j}}(\B{x})$,
$\BS{\Lambda}_{L_{j}}(\B{y})$ with $\B{x},\B{y}\in\B{B}_{5NL_{j+1}}(\B{0})$ is bounded by $(10NL_{j+1}+1)^2<C(N)L_{j+1}^2$. We can apply the bound \eqref{eq:DS} and write
\[
\prob{\C{S}_j}\leq C(N)L_{j+1}^{2Nd}L_{j}^{-2p}\leq L_j^{-2p+2Nd\alpha}.
\]
Therefore,
\[
\Omega_k^{\bad}\leq\, L_k^{- 2p + 2Nd\alpha }\sum_{i\geq 0}\left(\frac{L_{k+i}}{L_k}\right)^{- 2p + 2Nd\alpha}.
\]
For $2p> 2Nd\alpha$ and $L_0\geq 2$ the claim  follows from the inequality
\[
\frac{L_{k+i}}{L_k}\geq \left[L_k^{\alpha^i-1}\right].\qedhere
\]
\end{proof}

\subsection{Centers of localization}    \label{ssec:step.2}

Denote by $\BS{\varPhi}_n = \BS{\varPhi}_n(\omega)$ the normalized eigenfunctions of $\B{H}(\omega)$, $\omega\in\Omega_1$, with corresponding eigenvalues $E_n = E_n(\omega)\in I$. For each $n$ we define a \textit{center of localization} for $\BS{\varPhi}_n$ as a point $\B{\hat{x}}\in\D{Z}^d$ such that
\begin{equation}\label{eq:loccenter}
\| \one_{\B{C}(\B{\hat{x}})} \, \BS{\varPhi}_n  \| =
\max_{\B{y}\in\D{Z}^{Nd} } \, \| \one_{\B{C}(\B{y})}\, \BS{\varPhi}_n   \|.
\end{equation}
Since $\|\BS{\varPhi}_n\|=1$, for any given $n$ such centers exist and their number is finite. We will assume that, for any eigenfunction $\BS{\varPhi}_n$, the centers of localization $\xhatna$, $a = 1, \ldots,\hat C(n)$, are enumerated in such a way that $| \xhat_{n,1} |=  \min_a | \xhat_{n,a} |$.

\begin{lemma}\label{lem:lco.center.S}
There exists $k_0$ large enough such that, for all $\B{u}\in\D{Z}^{Nd}$, $\omega\in\Omega_1$ and $k\geq k_0$, if $\B{\hat{x}}_{n,a} \in \B{B}_{L_k}(\B{u})$ then the box
$\BS{\varLambda}_{L_k }(\B{u})$ is $(m,E_n)$-{\rm S}.
\end{lemma}

\begin{proof}
Assume otherwise. Then from \eqref{eq:EDI} it follows that
\[
\|\one_{\B{C}(\xhatna)}  \BS{\varPhi}_n  \|
\leq C'\eul^{-mL_k} \| \one_{\BS{\varLambda}^{\out}_{L_k}( \B{u}) }
\BS{\varPhi}_n  \|.
\]
Since the number of cells in $\BS{\varLambda}^{\out}_{L_k}( \B{u})$ is bounded by $L_k^{Nd}$, we conclude that
\[
\|  \one_{\B{C}(\xhatna)}  \BS{\varPhi}_n  \|
\leq C'\eul^{-mL_k} L_k^{Nd} \cdot \max_{\B{y}\in \B{B}^{\out}_{L_k}( \B{u}) }
\, \| \one_{\B{C}(\B{y})} \BS{\varPhi}_n \|.
\]
If $k_0$ is large enough so that $ C'\eul^{-mL_k}  L_k^{Nd} < 1$ for $k\geq k_0$, the  above inequality contradicts the definition of  $\xhatna$ as center of localization.
\end{proof}

\subsection{Annular regions}   \label{ssec:step.3}

From now on we work with the integer $k_0$ from Lemma~\ref{lem:lco.center.S}. Given $k>k_0$, set:
\begin{equation}\label{eq:Omgood}
\Omega_k^{\good}=\Omega_1\setminus\Omega_k^{\bad}.
\end{equation}
Assume that $\omega\in\Omega_k^{\good}$. Let $\xhat_{n,a}$, $\xhat_{n,b}$ be two centers of localization for the same eigenfunction $\BS{\varPhi}_n$. It follows from the definition of the event $\Omega_k^{\good}$ that the cubes $\BS{\varLambda}_{L_{i}}(\xhat_{n,a})$ and $\BS{\varLambda}_{L_{i}}(\xhat_{n,b})$ with $i\geq k-1$ cannot be separable, since they must be $(m,E)$-{\rm S}. Further, by Lemma~\ref{lem:CondGeomSep}, if $L_0\geq\R{r}_1$ then any cube of the form $\BS{\varLambda}_{L_k}(\B{y})$ with $|\B{y}|>|\xhat_{n,1}|+5NL_k$ is separable from $\BS{\varLambda}_{L_k}(\xhat_{n,1})$; this also applies, of course, to any localization center $\B{y}=\xhat_{n,a}$ with $a>1$, provided that such centers exist for a given $n$. Since $\omega\in\Omega_k^{\good}$, for any eigenfunction $\BS{\varPhi}_n$ there is no center of localization $\xhat_{n,a}$ either outside the cube $\BS{\varLambda}_{|\xhat_{n,1}|+5NL_k}(\B{0})$ or inside $\BS{\varLambda}_{|\xhat_{n,1}|}(\B{0})$ (since $|\xhat_{n,1}|=\min_a|\xhat_{n,a}|$). In other words, within the event $\Omega_k^{\good}$, all centers of localization $\xhatna$ with a fixed value of $n$ are located in the annulus
\[
\BS{\varLambda}_{|\xhat_{n,1}|+5NL_k}(\B{0})\setminus\BS{\varLambda}_{|\xhat_{n,1}|}(\B{0})
\]
of width $5NL_k$ and of inner radius $|\xhat_{n,1}|$. This explains why, for our purposes, an eigenfunction $\BS{\varPhi}_n$ can be effectively ``labeled'' by a single localization center.

In other words, although in this paper we cannot rule out the possibility of existence of multiple centers of localization at arbitrarily large distances (depending on $\BS{\varPhi}_n$ through $|\xhatone|$), such centers do not contribute to a ``radial'' quantum transport -- away from the origin $\B{0}$ -- which might have lead to dynamical delocalization.

\begin{lemma}\label{lem:3.3}
Given $k>k_0$, there exists $j_0 \geq k$ large enough such that if $j\geq j_0$, $\omega\in\Omega_k^{\good}$ and $\xhatone\in\B{B}_{L_j}(\B{0})$ then
\[
\left\|\left(1 - \one_{\BS{\varLambda}_{L_{j+2}}(\B{0})}\right)\BS{\varPhi}_n\right\|\leq\frac{1}{4}.
\]
\end{lemma}

\begin{proof}
By Lemma~\ref{lem:CondGeomSep} (see also \eqref{eq:cond.sep.5NL}),
\[
\forall\,i\geq j+1,\ \forall\,\B{w}\in\D{Z}^{Nd}\setminus\BS{B}_{5NL_i}(\B{0}),\text{ the cubes $\BS{\varLambda}_{L_{i}}(\B{w})$ and $\BS{\varLambda}_{L_{i}}(\B{0})$ are  separable.}
\]
In addition, we take $j\geq k$, as suggested in the lemma.

Next, we divide the complement $\D{R}^{Nd}\setminus\BS{\varLambda}_{5NL_{j+2}}(\B{0})$ into annular regions
\begin{equation}\label{eq:annuli}
\B{M}_i(\B{0})
:= \BS{\varLambda}_{5NL_{i+1}}(\B{0}) \setminus
\BS{\varLambda}_{ 5NL_{i}}(\B{0}),
\quad i\geq j+2,
\end{equation}
and write
\[
\left\|\left(1 -
\one_{\BS{\varLambda}_{L_{j+2}}(\B{0})}\right)\BS{\varPhi}_n\right\| ^2
= \sum_{i\geq j+2} \| \one_{\B{M}_i(\B{0})}\BS{\varPhi}_n\|^2\leq
\sum_{i\geq j+2}\; \sum_{ \B{w} \in \B{M}_i(\B{0})} \| \one_{\B{C}(\B{w})}
\BS{\varPhi}_n\|^2.
\]
Furthermore, $\xhatone\in\B{B}_{L_j}(\B{0})\subset\B{B}_{L_{i-1}}(\B{0})$,
so that by Lemma~\ref{lem:lco.center.S},  the cube
$\BS{\varLambda}_{L_{i}}(\B{0})$ must be $(m,E_n)$-S. Therefore, by the definition of the event $\Omega_k^{\good}$, the cube $\BS{\varLambda}_{L_{i}}(\B{w})$ is $(m,E_n)$-NS. Applying Lemma~\ref{lem:EDI} to the cube $\BS{\varLambda}_{L_{i}}(\B{w})$ and to the cell $\B{C}(\B{w})$, we obtain
\[
\|\one_{\B{C}(\B{w})}\BS{\varPhi}_n \|^2 \leq\eul^{-2mL_{i}}.
\]
Since the volume $\big|\B{M}_i(\B{0})\big|$ of the annular region $\B{M}_i(\B{0})$ grows polynomially in $L_i$, the assertion of Lemma~\ref{lem:3.3} follows.
\end{proof}

\subsection{Bounds on concentration  of localization centers}
\label{ssec:step.4}

\begin{lemma}\label{lem:3.4}
There exists a constant $c_2\in (0,\infty)$ such that for $\omega\in\Omega_k^{\bad}$,
$j\geq k$,
\begin{equation}\label{eq:DD.7}
\card \left\{n:\xhatone\in \B{B}_{L_{j+1}}( \B{0}) \right\}
\leq c_2\, L_{j+1}^{\alpha Nd}.
\end{equation}
\end{lemma}

\begin{proof}
The left-hand-side of \eqref{eq:DD.7} is nondecreasing in $j$, so we can restrict ourselves to the case $j\geq j_0$. First, observe that, with $\BS{\varLambda}_{L_{j+2}}=\BS{\varLambda}_{L_{j+2}}(\B{0})$
\begin{equation}\label{eq:DD.8}
\sum_{n:\,\xhatone\in \B{B}_{L_{j+1}}(\B{0})}\left( \one_{\BS{\varLambda}_{L_{j+2}}}P_I \one_{\BS{\varLambda}_{L_{j+2}}}\BS{\varPhi}_n,\BS{\varPhi}_n \right)\leq\tr\bigl(\one_{\BS{\varLambda}_{L_{j+2}}}P_I\bigr).
\end{equation}
Each term in the above sum is not less than $1/2$. Indeed,
\begin{align}\label{eq:DD.9}
&\left(\one_{\BS{\varLambda}_{L_{j+2}}} P_I\one_{ \BS{\varLambda}_{L_{j+2}}} \BS{\varPhi}_n,\BS{\varPhi}_n \right)\notag\\
&\qquad= \left(\one_{\BS{\varLambda}_{L_{j+2}}} P_I \BS{\varPhi}_n,
\BS{\varPhi}_n \right)- \left(\one_{\BS{\varLambda}_{L_{j+2}}} P_I\bigl(1 - \one_{{\BS{\varLambda}_{L_{j+2}}}}\bigr)\BS{\varPhi}_n, \BS{\varPhi}_n \right)\notag\\
&\qquad\geq\Big(\one_{\BS{\varLambda}_{L_{j+2}}} \BS{\varPhi}_n,\BS{\varPhi}_n\Big)-\frac{1}{4}\qquad\qquad\text{(using Lemma~\ref{lem:3.3})}\\
\label{eq:DD.10}
&\qquad= ( \BS{\varPhi}_n,\BS{\varPhi}_n) -
\left(\bigl(1 - \one_{\BS{\varLambda}_{L_{j+2}}}\bigr)
\BS{\varPhi}_n, \BS{\varPhi}_n\right) -\frac{1}{4}\notag\\
&\qquad\geq\frac{1}{2}\,.
\end{align}
Substituting the lower bounds from \eqref{eq:DD.9} --
\eqref{eq:DD.10} under the trace in Eqn \eqref{eq:DD.8},
we get the desired upper bound on the LHS of Eqn \eqref{eq:DD.7}.
\end{proof}

\subsection{Bounds for ``good'' samples of potential}\label{ssec:step.5}

\begin{lemma}\label{lem:3.5}
There exists an integer $k_1 = k_1(L_0)$ such that $\forall$ $k\geq k_1$, $\omega\in\Omega_{k+1}^{\good}$ and $\B{x}$ from the annular region $\B{M}_{k+1}$ defined in \eqref{eq:annuli},
\begin{equation}\label{eq:DD.11}
\left\| \one_{\BS{\varLambda}_{L_k}(\B{x})}\,P_I\,\xi (\B{H})\,
\one_{\BS{\varLambda}_{L_k}(\B{0})}\right\|
\leq\eul^{-m L_k/2} \| \xi \|_\infty.
\end{equation}
\end{lemma}

\begin{proof}
It suffices to prove \eqref{eq:DD.11} in the particular case where
$\| \xi\|_\infty\leq 1$, which we assume below. First, we bound
the LHS of \eqref{eq:DD.11} as follows:
\begin{align}\label{eq:DD.12}
\| \one_{\BS{\varLambda}_{L_k}(\B{x})} \,P_I\,\xi (\B{H})
\,\one_{\BS{\varLambda}_{L_k}(\B{0})}  \|
&\leq \sum_{E_n\in I} \, | \xi(E_n)| \, \| \one_{\BS{\varLambda}_{L_k}
(\B{x})} \BS{\varPhi}_n \| \, \| \one_{\BS{\varLambda}_{L_k}( \B{0})}
\BS{\varPhi}_n \|\notag\\
&\leq
\sum_{E_n\in I} \, \| \one_{\BS{\varLambda}_{L_k}(\B{x})} \BS{\varPhi}_n \|\,
\|\one_{\BS{\varLambda}_{L_k}(\B{0})} \BS{\varPhi}_n\|
\end{align}
since $\|\eta\|_\infty\leq 1$. Now divide the sum according to
where $\xhatone$ are located and write
\begin{align*}
\sum_{E_n\in I} \, \| \one_{\BS{\varLambda}_{L_k}(\B{x})}
\BS{\varPhi}_n \| \,
 \| \one_{\BS{\varLambda}_{L_k}( \B{0})} \BS{\varPhi}_n \|
&= \sum_{ \substack{ E_n\in I\\ \xhatone\in
{ \BS{\varLambda}_{5NL_{k+1}}( \B{0}) } }} \,
\| \one_{\BS{\varLambda}_{L_k}(\B{x})}  \BS{\varPhi}_n \| \,
\| \one_{\BS{\varLambda}_{L_k}(\B{0})} \BS{\varPhi}_n \|
\\ \\
&\quad
+ \sum_{j = k+1}^{\infty}  \sum_{ \substack{ E_n\in I\\ \xhatone\in
\B{M}_{ j }
(\B{0})} } \,
\| \one_{\BS{\varLambda}_{L_k}(\B{x})}  \BS{\varPhi}_n \| \,
\| \one_{\BS{\varLambda}_{L_k}(\B{0})} \BS{\varPhi}_n \|,
\end{align*}
with $\B{M}_i(\B{0})$ defined in \eqref{eq:annuli}. Since $\B{x}\in \B{M}_{k+1}(\B{0})$, we have $\B{B}_{L_k}( \B{x}) \cap \B{B}_{L_k}( \B{0}) = \varnothing$.
Then, by Lemma~\ref{lem:CondGeomSep}, the two cubes $\B{B}_{L_k}(\B{x})$
and $\B{B}_{L_k}(\B{0})$ are separable. In turn, this implies that one of these cubes is $(m,E_n)$-NS. Therefore, by Lemma~\ref{lem:3.4},
\[
\sum_{\substack{ E_n\in I\\  \xhatone\in
{ \BS{\varLambda}_{L_{k+1}}(\B{0}) }}} \, \|
\one_{\BS{\varLambda}_{L_k}(\B{x})} \BS{\varPhi}_n \|
\, \| \one_{\BS{\varLambda}_{L_k}(\B{0})} \BS{\varPhi}_n \|
\leq c_2  \, C' \, L_{k+1}^{\alpha N d} \,\eul^{-m L_k}.
\]
Furthermore, for $k>k_0$ large enough,
\begin{equation}\label{eq:DD.13}
\sum_{\substack{ E_n\in I\\ \xhatone\in
{\BS{\varLambda}_{L_{k+1}}(\B{0}) } }} \,
\|\one_{\BS{\varLambda}_{L_k}(\B{x})} \BS{\varPhi}_n \| \,
\| \one_{\BS{\varLambda}_{L_k}(\B{0})} \BS{\varPhi}_n \|
\leq \frac{1}{2}\eul^{-m L_k/2}.
\end{equation}
For any  $j\geq k+ 2$ and $\xhatone\in\B{M}_j(\B{0})$,
by Lemma~\ref{lem:lco.center.S}, the cube $\B{B}_{L_j}(\xhatone)$
must be $(m,E_n)$-S, so that $\B{B}_{L_j}(\B{0})$ has to be $(m,E_n)$-NS:
\[
\| \one_{\BS{\varLambda}_{L_k}(\B{0})} \BS{\varPhi}_n\|
\leq \| \one_{\BS{\varLambda}_{L_j}(\B{0})} \BS{\varPhi}_n\| \leq C'\eul^{-m L_j}.
\]
Applying again Lemma~\ref{lem:3.4}, we see that, if
$k\geq k_1$, then
\begin{align*}
\sum_{j=k+1}^{\infty}
\sum_{\substack{E_n\in I\\ \xhatone\in \B{M}_j (\B{0}) } }
\| \one_{\BS{\varLambda}_{L_k}(\B{x})} \BS{\varPhi}_n\|
\, \| \one_{\BS{\varLambda}_{L_k}(\B{0})} \BS{\varPhi}_n\|
&\leq C \sum_{j=k+ 2 }^{\infty}\eul^{-m L_j} L_j^{\alpha
N d}\\
&\leq \frac{1}{2}\eul^{-m L_k/2}.
\end{align*}
Combining this estimate with \eqref{eq:DD.12} and \eqref{eq:DD.13}, the
assertion of Lemma~\ref{lem:3.5} follows.
\end{proof}

\subsection{Bounds for ``bad'' samples of potential} \label{ssec:step.6}

\begin{lemma}
Let $k_1$ be as in Lemma~\ref{lem:3.5} and assume that $k\geq k_1$ and $\B{x}\in\B{M}_{k+1}(\B{0})$. We have:
\[
\expect\left[
\left\| \one_{\BS{\varLambda}_{L_k}(\B{x})}\, P_I\,\xi (\B{H})\,
\one_{\BS{\varLambda}_{L_k}(\B{0})}\right\|\right]
\leq \| \xi \|_\infty
\left( C L_k^{-2p  + 2Nd\alpha} +\eul^{-m L_k/2}\right).
\]
\end{lemma}

\begin{proof}
We again assume $\| \xi\|_\infty\leq 1$.
For $\omega\in\Omega_k^{\bad}$ we can use Sect.~\ref{ssec:step.1}
while for $\omega\in\Omega_k^{\good}$ we can use Sect.~\ref{ssec:step.5}.
More precisely, the above expectation is bounded by
\[
\prob{\Omega_k^{\bad}}
+\eul^{-m L_k/2}\,\prob{\Omega_k^{\good}}
\leq  C L_k^{-2p  + 2Nd\alpha} +\eul^{-m L_k/2}.\qedhere
\]
\end{proof}

\subsection{Conclusion}       \label{ssec:step.4}

For a compact subset $\B{K}\subset \D{R}^{Nd}$ we find an integer $k\geq k_1$ such that
$\B{K}\subset \BS{\varLambda}_{L_{ k }}(\B{0})$. Then, with the
annular regions $\B{M}_{j}(\B{0})$,
\begin{align*}
\expect\left[  \| \B{X}^Q\, P_I\,\xi (\B{H}(\omega))\, \one_{\B{K}} \|  \right]
&\leq L_k^Q+ \sum_{j\geq k+1 }\expect\left[\| \B{X}^Q \,\one_{\B{M}_{j}(\B{0})} P_I\,\xi (\B{H})\, \one_{\B{K}} \|  \right]\\
&\leq L_k^Q +\sum_{j\geq k+1 }L_{j+1}^Q\Biggl(\,\sum_{\B{w}\in\B{M}_j(\B{0})}\expect \left[ \| \one_{\BS{\varLambda}_{L_k}(\B{w})} P_I\,\xi (\B{H})\,\one_{\BS{\varLambda}_{L_k}(\B{0})} \|  \right]\Biggr)\\
&\leq L_k^Q + \sum_{j\geq k+1 } L_j^{\alpha Q} L_j^{Nd\alpha}
\left( L_j^{-2p + 2Nd\alpha } +\eul^{-m L_j/2} \right)< \infty,
\end{align*}
since $2p > 3Nd\alpha +\alpha Q$ by assumption \eqref{eq:DD.2},  and
$L_j\sim\left[ L_0^{\alpha^j}\right]$ grow fast enough.

This completes the proof of dynamical localization.\hfill\qed


\begin{bibdiv}
\begin{biblist}
\bib{AW09a}{article}{
   author={Aizenman, M.},
   author={Warzel, S.},
   title={Localization bounds for multiparticle systems},
   journal={Comm. Math. Phys.},
   volume={290},
   date={2009},
   number={3},
   pages={903--934},
}

\bib{AW09b}{misc}{
    author={Aizenman, M.},
   author={Warzel, S.},
   title={Complete dynamical localization in disordered quantum
multi-particle systems},
   status={arXiv:math-ph/0909:5432 (2009)},
   date={2009},
   pages={},
}

\bib{BCSS10a}{article}{
   author={Boutet de Monvel, A.},
   author={Chulaevsky, V.},
   author={Stollmann, P.},
   author={Suhov, Y.},
   title={Wegner-type bounds for a multi-particle continuous Anderson
   model with an alloy-type external potential},
   journal={J. Stat. Phys.},
   volume={138},
   date={2010},
   number={4-5},
   pages={553--566},
}

\bib{BCSS10b}{misc}{
   author={Boutet de Monvel, A.},
   author={Chulaevsky, V.},
   author={Stollmann, P.},
   author={Suhov, Y.},
   title={Anderson localization for a multi-particle alloy-type model},
   status={arXiv:math-ph/1004.1300 (2010)},
   date={2010},
}

\bib{CS08}{article}{
   author={Chulaevsky, V.},
   author={Suhov, Y.},
   title={Wegner bounds for a two-particle tight binding model},
   journal={Comm. Math. Phys.},
   volume={283},
   date={2008},
   number={2},
   pages={479--489},
}
\bib{CS09a}{article}{
   author={Chulaevsky, V.},
   author={Suhov, Y.},
   title={Eigenfunctions in a two-particle Anderson tight binding model},
   journal={Comm. Math. Phys.},
   volume={289},
   date={2009},
   number={2},
   pages={701--723},
}
\bib{CS09b}{article}{
   author={Chulaevsky, V.},
   author={Suhov, Y.},
   title={Multi-particle Anderson localisation: induction on the number of
   particles},
   journal={Math. Phys. Anal. Geom.},
   volume={12},
   date={2009},
   number={2},
   pages={117--139},
}


\bib{KZ03}{misc}{
   author={Klopp, F.},
   author={Zenk, H.},
   title={The integrated density of states for an interacting
          multielectron homogeneous model},
   status={arXiv:math-ph/0310031},
   date={2003},
}

\bib{St01}{book}{
   author={Stollmann, P.},
   title={Caught by disorder},
   series={Progress in Mathematical Physics},
   volume={20},
   note={Bound states in random media},
   publisher={Birkh\"auser Boston Inc.},
   place={Boston, MA},
   date={2001},
   pages={xviii+166},
}
\end{biblist}
\end{bibdiv}
\end{document}